\documentclass[prl,twocolumn,showpacs]{revtex4-2}

\usepackage{amsmath}
\usepackage{latexsym}
\usepackage{amssymb}
\usepackage{graphicx}
\usepackage[colorlinks=true, citecolor=blue, urlcolor=blue]{hyperref}
\usepackage{float}
\usepackage{amsfonts}
\usepackage{textcomp}
\usepackage{mathpazo}
\usepackage{comment}
\usepackage{xr}

\usepackage{bbm}

\usepackage{xcolor}
\definecolor{myurlcolor}{rgb}{0,.5,.5}
\definecolor{mycitecolor}{rgb}{0,.6,0}
\definecolor{myrefcolor}{rgb}{2,0,0}
\usepackage{hyperref}
\hypersetup{colorlinks,
linkcolor=myrefcolor,
citecolor=mycitecolor,
urlcolor=myurlcolor}

\usepackage[draft]{fixme}
\usepackage{amsmath,bbm}
\usepackage{graphicx}
\usepackage{amsfonts}
\usepackage{amssymb}
\usepackage{amsmath, amssymb, amsthm,verbatim,graphicx,bbm}
\usepackage{mathrsfs}
\usepackage{color,xcolor,longtable}

\usepackage{xr}
\makeatletter
\newcommand*{\addFileDependency}[1]{
  \typeout{(#1)}
  \@addtofilelist{#1}
  \IfFileExists{#1}{}{\typeout{No file #1.}}
}
\makeatother

\newcommand*{\myexternaldocument}[1]{
    \externaldocument{#1}
    \addFileDependency{#1.tex}
    \addFileDependency{#1.aux}
}

\myexternaldocument{manuscript_PRL}

\newcommand{\beq}[0]{\begin{equation}}
\newcommand{\eeq}[0]{\end{equation}}

\newcommand{\one}{\leavevmode\hbox{\small1\normalsize\kern-.33em1}}

\def\be{\begin{equation}}
\def\ee{\end{equation}}
\def\ben{\begin{eqnarray}}
\def\een{\end{eqnarray}}
\def\eea{\end{array}}
\def\bea{\begin{array}}

\newcommand{\Tr}[1]{\mathrm{Tr}#1}
\newcommand{\bei}{\begin{itemize}}
\newcommand{\eei}{\end{itemize}}
\newcommand{\ket}[1]{|#1\rangle}
\newcommand{\bra}[1]{\langle#1|}

\newcommand{\proj}[1]{\ket{#1}\!\!\bra{#1}}

\newcommand{\I}{\mathbbm{1}}

\newcommand{\p}{\Vec{p}}

\renewcommand{\emph}[1]{\textbf{#1}}

\makeatletter
\newtheorem*{rep@theorem}{\rep@title}
\newcommand{\newreptheorem}[2]{%
\newenvironment{rep#1}[1]{%
 \def\rep@title{#2 \ref{##1}}%
 \begin{rep@theorem}}%
 {\end{rep@theorem}}}
\makeatother

\theoremstyle{plain}
\newtheorem{thm}{Theorem}
\newtheorem*{thm*}{Theorem}
\newreptheorem{thm}{Theorem}
\newtheorem{lem}{Lemma}
\newtheorem{fakt}{Fact}

\newtheorem{defn}[thm]{Definition}

\theoremstyle{definition}

\theoremstyle{remark}

\usepackage[T1]{fontenc}

\begin{document}

\title{Operationally independent events can influence each other in quantum theory}
\author{Shubhayan Sarkar}
\email{shubhayan.sarkar@ulb.be}
\affiliation{Laboratoire d’Information Quantique, Université libre de Bruxelles (ULB), Av. F. D. Roosevelt 50, 1050 Bruxelles, Belgium}

\begin{abstract}	
In any known description of nature, two physical systems are considered independent of each other if any action on one of the systems does not change the other system. From our classical intuitions about the world, we further conclude that these two systems are not affecting each other in any possible way, and thus these two systems are causally disconnected or they do not influence each other. Building on this idea, we show that in quantum theory such a notion of ``classical independence'' is not satisfied, that is, two quantum systems can still influence each other even if any operation on one of the systems does not create an observable effect on the other. For our purpose, we consider the framework of quantum networks and construct a linear witness utilizing the Clauser-Horne-Shimony-Holt inequality. We also discuss one of the interesting applications resulting from the maximal violation of ``classical independence'' towards device-independent certification of quantum states and measurements.
\end{abstract}


\maketitle

{\it{Introduction---}} 
Nonlocality is one of the most fascinating aspects of quantum theory, encapsulating the absence of a local description for spatially separated quantum systems that can not communicate with each other. Discovered by Bell in 1964 \cite{Bell, Bell66} and then experimentally observed in the last decades \cite{Bellexp1, Bellexp2, Bellexp3, Bellexp4}, the phenomenon of nonlocality clearly establishes the departure of the quantum world from classical physics. An equivalent way to understand it is that two quantum systems can be correlated in a stronger way than two classical systems.


In this work, we pose an even more stringent inquiry: consider two systems that exhibit no correlation with each other, meaning they are mutually independent. The fundamental question we address is whether these two independent systems can influence each other in any manner. Equivalently, can one system exert an impact on its counterpart when there is no correlation and no communication between them? Drawing upon our classical understanding of the natural world, it can be logically deduced that in the absence of communication and with both systems being mutually independent, they cannot exert any influence on each other in any conceivable manner. We consider this viewpoint as a notion of classicality and term it ``classical independence''.

Here, we show that the notion of classical independence is violated in quantum theory, that is, two mutually independent quantum systems might affect each other if they are individually entangled to some different quantum systems. For our purpose, we consider the framework of quantum networks, in particular, the quantum bilocality scenario \cite{pironio1, Pironio22} with weaker constraints on the network. We then derive a linear inequality inspired by Clasuer-Horne-Shimony-Holt (CHSH) inequality \cite{CHSH}. Restricting to operationally independent correlations, we find an upper bound for correlations that can be described in a classically independent way. We then find a set of quantum states and measurements that violate this bound. For a remark, just like the bilocality scenario \cite{pironio1, Pironio22} our setup is inspired by the phenomenon of entanglement swapping. In fact, one can also consider this work as probing the nature of non-classicality that is violated in the entanglement swapping experiment. Furthermore, we find the maximum value of the inequality that can be attained in quantum theory. Interestingly using the methods presented in \cite{sarkar2023,sarkar2023universal}, also allows us to certify the quantum states and measurements in a device-independent way from the maximal violation of the constructed inequality. Most of the recent works in network nonlocality point to the fact that observing genuine quantum nonlocality in networks requires non-linear inequalities, which additionally require the assumption that the sources generate independent and identically distributed random variables, commonly known as the i.i.d. assumption. We show here that one can also obtain linear inequalities to observe network nonlocality by considering weaker assumptions on the underlying hidden variable models.

\begin{figure*}[t]
    \centering
    \includegraphics[scale=.7]{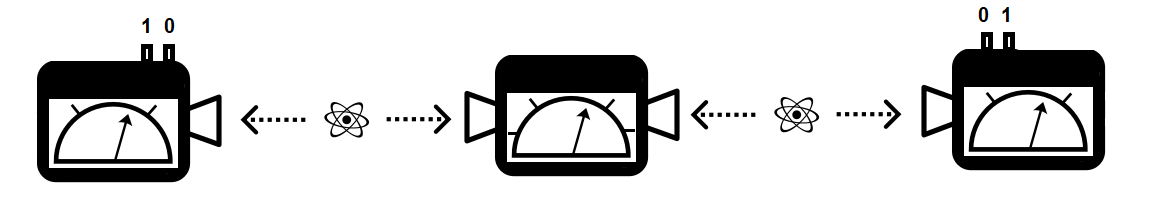}
    \caption{Weak-bilocality scenario. Alice and Bob each receive a single particle from their respective sources, which might be correlated to each other, while Eve receives two particles, one from each of these sources. Alice and Bob have the freedom to independently select and conduct two dichotomic measurements on their respective particles. In contrast, Eve's measurement is constrained to a single four-outcome measurement. None of the parties can communicate with each other.}
    \label{fig1}
\end{figure*} 

{\it{Classical independence.---}} Consider two systems with Alice and Bob such that measurements $\mathcal{A},\mathcal{B}$ with outcomes $a,b$ can be performed on them respectively. Now, we define when these two systems can be considered to be operationally independent of each other. 
\begin{defn}[Operational independence]\label{def1} Two systems are operationally independent if the probability of obtaining an outcome when performing a measurement on one system is completely independent of the other, that is,  
\begin{equation}\label{oip}
    p(a|\mathcal{A},b,\mathcal{B})=p(a|\mathcal{A})\qquad\forall a,b,\mathcal{A},\mathcal{B}
\end{equation}
    
\end{defn}
and similarly for Bob's outcome $b$.
The resulting joint probability $p(a,b|\mathcal{A},\mathcal{B})$ factors out using Bayes rule as
\begin{equation}\label{mu}
    p(a,b|\mathcal{A},\mathcal{B})=p(a|\mathcal{A})p(b|\mathcal{B})\qquad  \forall a,b,\mathcal{A},\mathcal{B}.
\end{equation}
Inspired by the above definition, we define the principle of no influence as stated below.
\begin{defn}[No-influence principle] \label{def2}Two systems do not influence each other if given any additional information $e$, there exist some hidden variables $\lambda$ such that the probability of obtaining an outcome when performing a measurement on one system is completely independent of the other, that is,  
\begin{equation}\label{nip}
    p(a|\mathcal{A},b,\mathcal{B},\lambda,e)=p(a|\mathcal{A},\lambda,e)\qquad\forall a,b,e,\mathcal{A},\mathcal{B},\lambda.
\end{equation}
\end{defn}
and similarly for Bob's outcome $b$.
It is straightforward to observe that no-influence principle implies operational independence. This brings us to the definition of classicality, which we call classical independence.

\begin{defn}[Classical independence] Two operationally independent events [def. \ref{def1}] are classically independent of each other if they do not influence each other [def. \ref{def2}], or to put it simply the notion of classical independence means
\begin{eqnarray}
    \mathrm{Operational\ Independence}\implies  \mathrm{No\ Influence}.
\end{eqnarray}
\end{defn}
Equivalently the above definition can be understood as, if the correlations between two parties are mutually independent for any possible choice of measurement of both parties, then given any additional information $e$ there always exists a hidden variable model where both parties do not influence each other. 

Consider again the no-influence principle which is mathematically equivalent to the assumption of local causality in the Bell scenario. However, the striking difference is the fact that the parties involved in the Bell scenario are not operationally independent. Furthermore, it should be noted that Alice's or Bob's results can be affected by some Eve's results, who is not locally causal to either Alice or Bob. As a result, we name the assumption as the "no-influence principle" to signify that there is no causal connection between Alice and Bob but each of them could be influenced by some other spatially separated system. Furthermore, in the context of Bell scenario, one usually justifies the assumption of local causality due to space-like separated events and postulates of relativity. On the contrary, here it is more natural as Alice and Bob are not correlated to each other. Consequently, the scenario considered in this work is weaker when compared to the Bell scenario, that is, we identify non-classical behaviour even in situations where one can not violate a Bell inequality. Unlike the Bell scenario, it should be noted here that to ensure operational independence \eqref{oip}, one needs to perform "all" possible measurements. Although not practical at present, there exists an operational criterion to ensure it, unlike the Bell scenario where one can not justify special relativity operationally.

Let us now construct a scenario where we can observe the violation of classical independence with quantum states and measurements. A natural scenario that one could investigate in this regard is the standard Bell scenario. However, it is quite clear that if the correlations between two parties are operationally independent \eqref{mu}, then one can never observe any violation of a Bell inequality. Consequently in this work, we consider the bilocality scenario \cite{pironio1} with three parties as described below.

\textit{The scenario---} 
We consider three parties namely, Alice, Bob and Eve in three different spatially separated labs. Alice and Bob receive a single particle from sources $S_1, S_2$ respectively and Eve receives two particles from both the sources. Unlike the bilocality scenario, the sources here need not be independent of each other, thus we call it ``weak-bilocality scenario''. Now, Alice and Bob perform two dichotomic measurements on their particles which they can freely choose. Eve on the other hand can only perform a single four-outcome measurement. The measurement inputs of Alice and Bob are denoted as $x,y=0,1$ respectively and their outcomes are denoted as $a,b=0,1$, whereas the outcomes of Eve are denoted as $e=0,1,2,3$. The scenario is depicted in Fig. \ref{fig1}. In Fig. \ref{fig2}, we show the causality graph of the weak-bilocality scenario.

The experiment is repeated enough times to construct the joint probability distribution or correlations, $\vec{p}=\{p(a,b,e|x,y)\}$ where $p(a,b,e|x,y)$ denotes the probability of obtaining outcome $a,b,e$ by Alice, Bob and Eve when they choose the inputs $x,y$ respectively. These probabilities can be computed in quantum theory using the Born rule as
\begin{eqnarray}
p(a,b,e|x,y)=\Tr\left[(N^A_{a|x}\otimes N^B_{b|y}\otimes N^E_e)\rho_{ABE}\right]
\end{eqnarray}
where $N^A_{a|x},N^B_{b|y},N^E_e$ denote the measurement elements of Alice, Bob and Eve corresponding to $x,y$ input and $\rho_{ABE}$ denotes the joint state generated by the source $S_1,S_2$. The measurement elements are positive semi-definite and $\sum_aN^A_{a|x}=\sum_bN^B_{b|y}=\sum_eN^E_e=1$ for all $x,y$.
It is important to recall here that Alice and Bob can not communicate with each other during the experiment.

It is usually simpler to express the probabilities in terms of expectation values as
\begin{eqnarray}\label{exp1}
\langle\mathcal{A}_x\mathcal{B}_yN^E_e\rangle=p(0,0,e|x,y)+p(1,1,e|x,y)\nonumber\\-p(0,1,e|x,y)-p(1,0,e|x,y)
\end{eqnarray}
where $\mathcal{A}_x, \mathcal{B}_y$ denote Alice's and Bob's observable corresponding to the input $x,y$ respectively and can be expressed as $s_i=N^s_{0|i}-N^s_{1|i}$ for $s=A,B$. 

{\it{Violation of classical independence.---}}
Let us now restrict that the correlations $p(a,b|x,y)$ are operationally independent [def. \ref{def1}], that is, $p(a,b|x,y)=p(a|x)p(b|y)$.

Now, let us express the joint probability distribution $p(a,b,e|x,y)$ as 
\begin{equation}
    p(a,b,e|x,y)=\sum_{\lambda}p(\lambda) p(a,b,e|x,y,\lambda)
\end{equation}
Using Bayes rule, the probability $p(a,b,e|x,y,\lambda)$ can be rewritten as
\begin{equation}
p(a,b,e|x,y,\lambda)=p(a|x,b,y,e,\lambda)p(b|x,y,e,\lambda)p(e|x,y,\lambda).
\end{equation}
Assuming no-influence [def. \ref{def2}] allows us to conclude that $p(a|x,b,y,e,\lambda)=p(a|x,e,\lambda)$ and $p(b|x,y,e,\lambda)=p(b|y,e,\lambda)$. Furthermore,  the fact that Eve's outcome $e$ is independent of Alice's or Bob's inputs $x,y$ allows us to conclude that
\begin{equation}\label{prob1}
p(a,b,e|x,y)=\sum_{\lambda}p(a|x,e,\lambda)p(b|y,e,\lambda)p(e|\lambda)p(\lambda).
\end{equation}
Notice that in the bilocality scenario \cite{pironio1}, 
 one additionally assumes that $p(a|x,e,\lambda)=p(a|x,\lambda)$ and $p(b|y,e,\lambda)=p(b|y,\lambda)$.

Let us now consider an example, as suggested by the referee, exhibiting a classical implementation of the above-presented notion. Consider again the weak-bilocality scenario [see Fig. \ref{fig1}] such that each source distributes one bit to each party. Furthermore,  both bits of Alice and Eve are either $0$ or $1$ with equal probability. Similarly, both bits of Bob and Eve are either $0$ or $1$ with equal probability. In this case, Alice and Bob are operationally independent, since each gets in half of the cases when the bit is $0$ or $1$ independent of the other party. The additional information $e$ here is whether Eve's both bits are equal or different, that is, $e=$ "$0$ (equal)" or "$1$ (different)". This is similar to the scheme of classical teleportation or entanglement swapping \cite{Spekkens2}. Now, one can always construct a hidden variable model with $\lambda=0,1$ such that for any  $e,a,b=0,1$ as
\begin{eqnarray}
    p(a|e,\lambda)=\delta_{a,\lambda},\quad p(b|e,\lambda)=\delta_{b\oplus e,\lambda}
\end{eqnarray}
such that no-influence principle \eqref{nip} is satisfied. Thus, Alice and Bob are classically independent.
 
Inspired by \cite{CHSH}, we will now construct a linear functional constructed from the joint probability distribution $\vec{p}$ to show that quantum theory violates classical independence. In terms of observables, the functional can be represented as
\begin{eqnarray}\label{Wit1}
    \mathcal{I}=\big\langle\mathcal{A}_0(\mathcal{B}_0-\mathcal{B}_1)\mathcal{E}_0+\mathcal{A}_1(\mathcal{B}_0+\mathcal{B}_1)\mathcal{E}_1\big\rangle
\end{eqnarray}
where $\mathcal{E}_0=N_0^E-N_1^E-N_2^E+N_3^E$ and $\mathcal{E}_1=N_0^E+N_1^E-N_2^E-N_3^E$. As shown below, the above inequality can be broken up into conditional CHSH inequalities, up to the presence of Eve's measurement, which were useful to prove that every pure entangled state is Bell nonlocal \cite{popescu} and self-testing the Bell basis \cite{Marco}. 

Let us compute the maximum value of $\mathcal{I}$ \eqref{Wit1} achievable using correlations that satisfy ``classical independence". We will further call this value as the
classical bound and denote it as $\beta_{C}$. For this purpose, let us express the expectation value \eqref{exp1} using \eqref{prob1} as
\begin{eqnarray}
\langle\mathcal{A}_x\mathcal{B}_yN^E_e\rangle=\sum_{\lambda}p(\lambda)p(e|\lambda)\big(p(0|x,e,\lambda)-p(1|x,e,\lambda)\big)\nonumber\\ \big(p(0|y,e,\lambda)-p(1|y,e,\lambda)\big)\quad
\end{eqnarray}
which can be simply stated as 
\begin{eqnarray}\label{CI1}
\langle\mathcal{A}_x\mathcal{B}_yN^E_e\rangle=\sum_{\lambda}p(\lambda)p(e|\lambda)\langle\mathcal{A}_{x,e,\lambda}\rangle\langle\mathcal{B}_{y,e,\lambda}\rangle.
\end{eqnarray}
Using the above expression, we now calculate the classical bound of $\mathcal{I}$ \eqref{Wit1}.

\begin{figure}
    \centering
    \includegraphics[scale=.27]{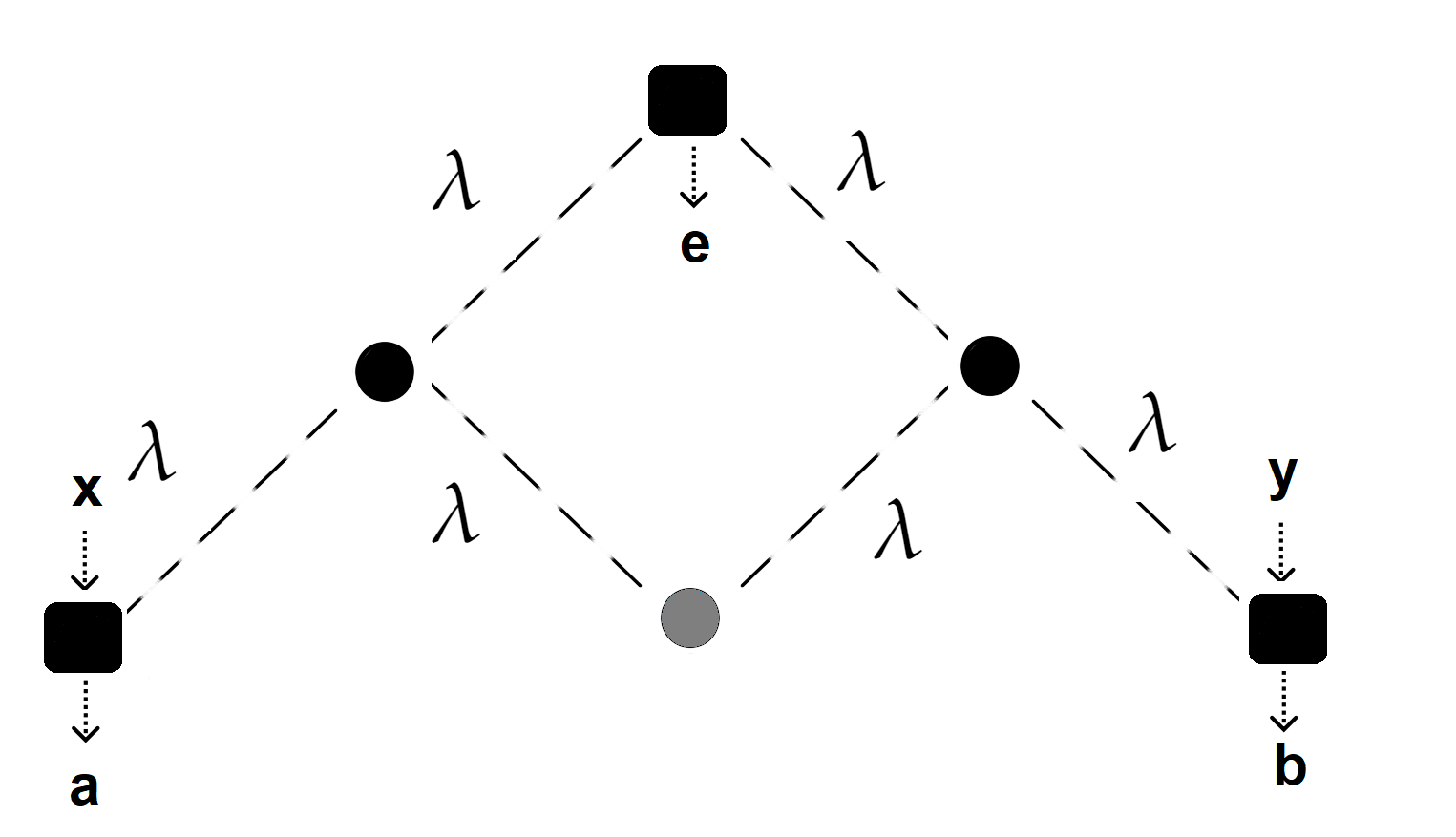}
    \caption{Causality graph of the weak-bilocality scenario. The square boxes represent the measurement devices and the circles represent the sources. The grey circle represents a hidden variable that might correlate the sources.}
    \label{fig2}
\end{figure}

\begin{fakt}\label{fact1}
    Consider the functional $\mathcal{I}$ \eqref{Wit1}. The maximum value $\beta_{C}$ that can be achieved by classically independent correlations is $\beta_{C}=2$.
\end{fakt}

\begin{proof}
    For our purpose, we rewrite the functional $\mathcal{I}$ \eqref{Wit1} as
    \begin{eqnarray}\label{wit2}
        \mathcal{I}= \sum_{i,j=0,1}\langle\mathcal{I}_{CHSH}^{i,j}N_{i,j}^{E}\rangle
    \end{eqnarray}
    such that $N_{i,j}^{E}\equiv N_{2i+j}^E$ and 
    \begin{equation}\label{wit3}
    \mathcal{I}_{CHSH}^{i,j}=(-1)^i(-1)^j\mathcal{A}_0(\mathcal{B}_0-\mathcal{B}_1)+\mathcal{A}_1(\mathcal{B}_0+\mathcal{B}_1)
    \end{equation}
    Now, taking the absolute value of Eq. \eqref{wit2} and using triangle inequality gives us
    \begin{eqnarray}
          |\mathcal{I}|\leq\sum_{i,j=0,1}\left|\langle\mathcal{I}_{CHSH}^{i,j}N_{i,j}^{E}\rangle\right|.
    \end{eqnarray}
    Expanding the terms on the right-hand side individually and using \eqref{CI1} gives us
    \begin{eqnarray}
    \left|\langle\mathcal{I}_{CHSH}^{i,j}N_{i,j}^{E}\rangle\right|\leq\sum_{\lambda}p(\lambda)p(i,j|\lambda) \Big( \Big|\langle\mathcal{A}_{0,i,j,\lambda}\rangle(\langle\mathcal{B}_{0,i,j,\lambda}\rangle-\nonumber\\ \langle\mathcal{B}_{1,i,j,\lambda}\rangle)\Big|+\Big|(-1)^j\langle\mathcal{A}_{1,i,j,\lambda}\rangle(\langle\mathcal{B}_{0,i,j,\lambda}\rangle+\langle\mathcal{B}_{1,i,j,\lambda}\rangle)\big\rangle\Big|\Big).\ \ 
    \end{eqnarray}
    Now, using the fact that
    $|\langle\mathcal{A}_{x,e,\lambda}\rangle|\leq 1$ for any $x,e,\lambda$ we obtain that
\begin{eqnarray}\label{CI2}
    \left|\langle\mathcal{I}_{CHSH}^{i,j}N_{i,j}^{E}\rangle\right|\leq\sum_{\lambda}p(\lambda)p(i,j|\lambda) \Big(\Big|\langle\mathcal{B}_{0,i,j,\lambda}\rangle-\langle\mathcal{B}_{1,i,j,\lambda}\rangle\Big|\nonumber\\ +\Big|\langle\mathcal{B}_{0,i,j,\lambda}\rangle+\langle\mathcal{B}_{1,i,j,\lambda}\big\rangle\Big|\Big).\quad
\end{eqnarray}
Furthermore, 
\begin{eqnarray}
    \Big|\langle\mathcal{B}_{0,i,j,\lambda}\rangle-\langle\mathcal{B}_{1,i,j,\lambda}\rangle\Big|+\Big|\langle\mathcal{B}_{0,i,j,\lambda}\rangle+ \langle\mathcal{B}_{1,i,j,\lambda}\big\rangle\Big|\nonumber\\=2\max(|\langle\mathcal{B}_{0,i,j,\lambda}\rangle|,|\langle\mathcal{B}_{1,i,j,\lambda}\rangle|)\leq 2.
\end{eqnarray}
Consequently, we obtain from Eq. \eqref{CI2} that 
\begin{eqnarray}
    \left|\langle\mathcal{I}_{CHSH}^{i,j}N_{i,j}^{E}\rangle\right|\leq 2 p(i,j)
\end{eqnarray}
which summing over $i,j$ finally gives us $|\mathcal{I}|\leq2$. This completes the proof.
\end{proof}

Consider now the quantum state $\ket{\psi_{ABE}}=\ket{\phi^+_{A\overline{A}}}\otimes\ket{\phi^+_{B\overline{B}}}$ such that $E\equiv \overline{AB}$ and $\ket{\phi^+}$ is the two-qubit maximally entangled state. It is easy to check that such a state will always generate correlations between Alice and Bob that are operationally independent [def. \ref{def1}]. Also, consider that Alice's and Bob's observables are given by
\begin{eqnarray}\label{abobs}
    \tilde{\mathcal{A}}_0&=&\sigma_z,\qquad\tilde{\mathcal{A}}_1=\sigma_x\nonumber\\
     \tilde{\mathcal{B}}_0&=&\frac{\sigma_z+\sigma_x}{\sqrt{2}},\qquad\tilde{\mathcal{B}}_1=\frac{\sigma_x-\sigma_z}{\sqrt{2}}
\end{eqnarray}
along with Eve's measurement given by the Bell-basis as $\tilde{N}_{i,j}^E=\ket{\phi_{i,j}}\!\bra{\phi_{i,j}}$ where $\ket{\phi_{i,j}}=\frac{1}{\sqrt{2}}(\ket{i\ j}+(-1)^{i}\ket{\overline{i}\ \overline{j}})$ where $i,j=0,1$ and $N_{i,j}^E\equiv N_{2i+j}^E$. Plugging these states and observables in the functional $\mathcal{I}$ \eqref{Wit1} gives us the value $2\sqrt{2}$. Thus, quantum theory violates the notion of classical independence. Consequently, one can conclude that systems that are operationally independent can influence each other. We show in theorem \ref{theorem1m} that $2\sqrt{2}$ is in fact the maximal value achievable using quantum theory.

We can identify some necessary conditions to violate the notion of classical independence using quantum states and measurements. The first necessary condition is that Eve needs to perform an entangled measurement as one can observe from Eq. \eqref{prob1}. This is contrary to the violation of bilocality which can also happen with separable measurements with Eve. Further on, the sources generating entangled states between Alice-Eve and Bob-Eve are necessary. Although the inequality \eqref{Wit1} considered in this work requires incompatible measurements to obtain any violation, one further needs to explore whether incompatible measurements are a necessity to violate classical independence or not. Let us now discuss an interesting application arising due to the violation of classical independence.

\textit{Self-testing.---} Self-testing is the strongest device-independent scheme that allows one to certify the quantum states and measurements without making any assumption on the devices involved apart from the fact that they are governed by quantum theory \cite{SupicReview}. Self-testing in quantum networks has been explored recently \cite{Marco, NLWEsupic, JW2, Allst1, supic4, sekatski, sarkar2023}, however, all of these schemes also assume that the sources are independent [see nevertheless \cite{sarkar15}]. Here, we do not need to assume it as we show below that the condition of operational independence [see Eq. \eqref{oip}] allows one to conclude that the sources are independent. 


\begin{lem}\label{lem1}
    Consider the weak-bilocality scenario such that the state shared between Alice and Bob is given by $\rho_{AB}$. If Alice and Bob are operationally independent, then $\rho_{AB}=\rho_{A}\otimes\rho_{B}$.
\end{lem}
\begin{proof} If Alice and Bob are operationally independent, then 
\begin{eqnarray}
p(a,b|\mathcal{A},\mathcal{B})=p(a|\mathcal{A})p(b|\mathcal{B})\qquad  \forall a,b,\mathcal{A},\mathcal{B}.
\end{eqnarray}
Putting quantum states and restricting to rank-one projective quantum measurements gives us
\begin{eqnarray}
    \bra{\psi^a}\bra{\psi^b}\rho_{AB}\ket{\psi^a}\ket{\psi^b}= \bra{\psi^a}\rho_{A}  \ket{\psi^a}\bra{\psi^b}\rho_{B}  \ket{\psi^b}
\end{eqnarray}
for all $\ket{\psi^a},\ket{\psi^b}$. Thus, we can conclude that
\begin{eqnarray}
\bra{\psi^a}\bra{\psi^b}\rho_{AB}\ket{\psi^a}\ket{\psi^b}= \bra{\psi^a}\bra{\psi^b}\rho_{A}\otimes\rho_{B}\ket{\psi^a}\ket{\psi^b}
\end{eqnarray}
and consequently we have that $\rho_{AB}=\rho_{A}\otimes\rho_{B}$. This completes the proof.
    
\end{proof}

Let us now state the self-testing result. Let us recall here that any measurement can only be certified on the local support of the states. Consequently, one can always assume that the local states of all the parties are full-rank. Furthermore, our self-testing result is based on the definitions introduced in \cite{sarkar2023universal}.

\setcounter{thm}{0}

\setcounter{thm}{0}
\begin{thm}\label{theorem1m}
Assume that the operationally independent correlations $\p$ attain the quantum bound of $\mathcal{I}$ \eqref{Wit1}.  
Then, for $s=A,B$; 
 (i) The Hilbert spaces of all the parties decompose as $\mathcal{H}_{s}=\mathcal{H}_{s'}\otimes\mathcal{H}_{s''}$  and $\mathcal{H}_{\overline{s}}=\mathcal{H}_{\overline{s}'}\otimes\mathcal{H}_{\overline{s}''}$, where  $\mathcal{H}_{s'}=\mathcal{H}_{\overline{s}'}\equiv\mathbb{C}^2$ is the target Hilbert space and $\mathcal{H}_{s''},\mathcal{H}_{\overline{s}''}$ denotes the junk Hilbert spaces; (ii) There exist local unitary transformations $U_{s}:\mathcal{H}_{s}\rightarrow\mathcal{H}_{s}$ and $U_{\overline{s}}:\mathcal{H}_{\overline{s}}\rightarrow\mathcal{H}_{\overline{s}}$ 
such that 
\begin{equation}\label{statest1}
(U_{s}\otimes U_{\overline{s}})\ket{\psi_{s\overline{s}}}=|\phi^+_{s'\overline{s}'}\rangle\otimes\ket{\xi_{s''\overline{s}''}}    
\end{equation}
for some junk state  $\ket{\xi_{s''\overline{s}''}}\in\mathcal{H}_{s''}\otimes\mathcal{H}_{\overline{s}''}$, and the measurements of all parties are certified as
\begin{equation}\label{meast1}
\overline{U}\,N_{i,j}^E\,\overline{U}^{\dagger}=\proj{\phi_{i,j}}_{E'}\otimes\I_{E''}, 
\end{equation}
where $\overline{U}=\otimes_s U_{\overline{s}}$ and $E\equiv \overline{AB}$ denoting the system of Eve such that $\mathcal{H}_E=\mathcal{H}_{\overline{A}}\otimes\mathcal{H}_{{\overline{B}}}=\mathbb{C}^2\otimes\mathbb{C}^2\otimes\mathcal{H}_{\overline{A}''}\otimes\mathcal{H}_{\overline{B}''}\equiv(\mathbb{C}^2\otimes\mathbb{C}^2)_{E'}\otimes\mathcal{H}_{E''}$ with,
\begin{eqnarray}
   &U_{B}\mathcal{B}_{0}U_{B}^{\dagger}=\frac{\sigma_z+\sigma_x}{\sqrt{2}}\otimes\I_{B''},\ \   U_{B}\mathcal{B}_{1}U_{B}^{\dagger}=\frac{\sigma_x-\sigma_z}{\sqrt{2}}\otimes\I_{B''},\nonumber\\
      &U_{A}\mathcal{A}_{0}U_{A}^{\dagger}=\sigma_z\otimes\I_{A''},\quad  U_{A}\mathcal{A}_{1}U_{A}^{\dagger}=\sigma_x\otimes\I_{A''}.
\end{eqnarray}

\end{thm}

\begin{proof}
    Let us consider the functional $\mathcal{I}$ written in the form in Eq. \eqref{wit2}. As shown in \cite{Bamps}, one can use the sum-of-squares decomposition to obtain that $\mathcal{I}_{CHSH}^{0,0}\leq2\sqrt{2}\I$ and similarly one can show that $\mathcal{I}_{CHSH}^{i,j}\leq2\sqrt{2}\I$.  As $N_{i,j}^{E}$ are positive, we can conclude from Eq. \eqref{wit2} that
    \begin{eqnarray}
     \mathcal{I}\leq2\sqrt{2}\sum_{i,j=0,1}\langle\I\otimes N_{i,j}^{E}\rangle\leq 2\sqrt{2}.
        \end{eqnarray}
    Consequently, one can conclude that the quantum bound of $\mathcal{I}$ is $2\sqrt{2}$.

    Now, if one attains the value $2\sqrt{2}$ of the functional $\mathcal{I}$, then we have that
    \begin{eqnarray}
         \mathcal{I}= \sum_{i,j=0,1}\langle\mathcal{I}_{CHSH}^{i,j}N_{i,j}^{E}\rangle=2\sqrt{2}
    \end{eqnarray}
    which can be rewritten as
    \begin{eqnarray}
    \sum_{i,j=0,1}\langle(2\sqrt{2}\I-\mathcal{I}_{CHSH}^{i,j})N_{i,j}^{E}\rangle=0.
    \end{eqnarray}
    Consider now the states $p(i,j|E)\rho_{AB}^{i,j}=\Tr_E(N_{i,j}^E\rho_{ABE})$ where $p(i,j|E)=\Tr(N_{i,j}^E\rho_{ABE})$ using which the above expression can be written as
    \begin{eqnarray}
    \sum_{i,j=0,1}p(i,j|E)\Tr[(2\sqrt{2}\I-\mathcal{I}_{CHSH}^{i,j})\rho_{AB}^{i,j}]=0.
    \end{eqnarray}
As we assume that the local state of Eve is full-rank and thus every outcome of Eve occurs due to which $p(i,j|E)\ne0$ for any $i,j$. Thus, we have that 
\begin{eqnarray}
\Tr(\mathcal{I}_{CHSH}^{i,j}\rho_{AB}^{i,j})=2\sqrt{2}\qquad \forall i,j.
\end{eqnarray}
Now, one can straightway adapt the proof presented in \cite{sarkar2023} to first self-test the states $\rho_{AB}^{i,j}$ to be
\begin{eqnarray}\label{st22}
   U_A\otimes U_B \rho_{AB}^{i,j}U_A^{\dagger}\otimes U_B^{\dagger}=\proj{\phi_{i,j}}\otimes\rho_{A''B''}^{i,j}
\end{eqnarray}
where $\ket{\phi_{i,j}}$ are given in the manuscript. The measurements of Alice and Bob are also certified as given in Eq. \eqref{meast1}.
Now, using Lemma 1 from the manuscript and then the above-certified states \eqref{st22} one can then self-test the states generated by the sources to be the maximally entangled state \eqref{statest1} which is then used to certify Eve's measurement \eqref{meast1}. It is important to notice that in Ref. \cite{sarkar2023}, one additionally assumes that the probabilities $p(i,j|E)$ are equal for all $i,j$. However, one can straightaway adapt the proof without such an assumption and get exactly the same conclusions.  This completes the proof.
\end{proof}

\textit{Discussions---} 
Let us observe that the assumptions considered in this work are weaker when compared to the bilocality scenario \cite{pironio1} as we allow Eve to influence Alice's and Bob's results. Furthermore, in any quantum network including the bilocality scenario one needs to further assume that the sources are statistically independent of each other [see nevertheless \cite{sarkar15}]. This is an extremely strong assumption as one can never operationally establish that two sources are independent of each other. However, here we do not put any restrictions on the sources and even allow them to be entangled. Furthermore, the bilocality scenario has already been experimentally implemented \cite{Sun} and thus we believe that the violation of inequality \eqref{Wit1} can be easily tested. 

Analyzing the above result from a realist perspective gives an interesting insight towards understanding whether a measurement on an entangled counterpart produces a physical change on the other. In the Bell scenario, a possible explanation of the observed nonlocal correlations is that the measurement by Alice updates the state with Bob or vice versa. However, such an explanation in the above-presented scenario is not consistent. First, consider that Eve performed her entangled measurement before Alice and Bob, then as the states between Alice and Bob are entangled then one can explain the violation of classical-independence using a similar realist explanation as the Bell scenario. However, consider now that Eve has not performed her entangled measurement, then as the states between Alice and Bob are separable any measurement by Alice should not alter Bob's state but can alter Eve's state. Consequently, there exist spacelike frames of reference where Alice's state update is caused by Bob's measurement and other frames where it remains unchanged. 
Thus, whether the "physical state" of Alice gets updated when Bob performs a measurement depends on the information about Eve's result which again is problematic if one considers that the cause-effect relationship is not epistemic.  

Several interesting problems follow up from our work. The most interesting among them would be to explore in detail whether a cause-effect relationship between two events is consistent in quantum theory or not. A simpler problem will be to extend the weak-bilocality scenario to the multipartite regime with arbitrary number of sources or higher number of outcomes. Furthemore, it will be extremely interesting if one can use the above-presented self-testing result to construct a device-independent key distribution scheme or for randomness certification.

{\it{Acknowledgments---}}We thank Stefano Pironio for insightful comments. This project was funded within the QuantERA II Programme (VERIqTAS project) that has received funding from the European Union’s Horizon 2020 research and innovation programme under Grant Agreement No 101017733.

\providecommand{\noopsort}[1]{}\providecommand{\singleletter}[1]{#1}%

\end{document}